\documentclass[runningheads]{llncs}
\usepackage[T1]{fontenc}
\usepackage{amsfonts}
\usepackage{amsmath}
\usepackage{amssymb}
\usepackage{appendix}
\usepackage{cancel}
\usepackage{hyperref}
\usepackage{tabularx}

\usepackage{graphicx}
\usepackage{booktabs}
\usepackage{algorithm}
\usepackage{algpseudocode}
\graphicspath{{figures/}}
\usepackage{color}

\usepackage{tikz}
\usetikzlibrary{positioning,calc}
\tikzset{
    factor/.style={draw, fill=black, minimum size=3.5mm, inner sep=0pt},
    equality/.style={draw, fill=white, minimum size=3.5mm, inner sep=0pt, font=\scriptsize},
    edgelabel/.style={font=\scriptsize, inner sep=1pt}
}
\usepackage{pgfplots}
\usepackage{circuitikz}
\usepackage{tikz-cd}
\usetikzlibrary{intersections}
\usepackage{bm}
\usetikzlibrary{positioning}
\usepgfplotslibrary{fillbetween}
\pgfplotsset{compat=1.5}
\usetikzlibrary{pgfplots.groupplots}
\usetikzlibrary{shapes.geometric,backgrounds}

\ctikzset{bipoles/length=.6cm}

\usepackage{tikz}
\usepackage{xifthen}

\pgfdeclarelayer{bg}    
\pgfsetlayers{bg,main}  
\definecolor{beige}{RGB}{245, 245, 220}

\definecolor{darkgrey}{RGB}{75, 75, 75}
\definecolor{lightgrey}{RGB}{250, 250, 250}

\usetikzlibrary{calc, arrows, fit, positioning, patterns, decorations.pathreplacing, shapes}
\tikzstyle{dash} = [dashed, -latex,>=latex]
\tikzstyle{line} = [draw, -latex,>=latex]
\tikzstyle{smallbox} = [draw, minimum size=5.0mm]
\tikzstyle{box} = [draw, minimum size=7.0mm]
\tikzstyle{bigbox} = [draw, minimum size=10.0mm]
\tikzstyle{rectangle} = [draw, minimum width=10.0mm, minimum height=20.0mm]
\tikzstyle{switch} = [trapezium, trapezium angle=120, draw, rotate=90,  inner ysep=5pt, outer sep=5pt,
minimum height=7mm, minimum width=7mm]
\tikzstyle{roundbox} = [draw, circle, inner sep=0pt, minimum size=3mm]
\tikzstyle{clamped} = [draw, fill=darkgrey, minimum size=0.15cm]
\tikzstyle{msgcircle} = [shape=circle, draw, inner sep=0pt, minimum size=4mm, fill=white, font=\scriptsize]
\tikzstyle{darkmsgcircle} = [shape=circle, draw, inner sep=0pt, minimum size=4mm, fill=darkgrey, text=white, font=\scriptsize]
\tikzstyle{redmsgcircle} = [shape=circle, draw=red, inner sep=0pt, minimum size=4mm, text=red, font=\scriptsize]
\tikzstyle{reddarkmsgcircle} = [shape=circle, draw=red, inner sep=0pt, minimum size=4mm, fill=red, text=white, font=\scriptsize]
\tikzstyle{msgdoublecircle} = [shape=circle, double, double distance=1.5pt, draw, inner sep=0pt, minimum size=5mm, fill=white]
\tikzstyle{darkmsgdoublecircle} = [shape=circle, double, double distance=1.5pt, draw, inner sep=0pt, minimum size=5mm, fill=darkgrey, text=white, font=\bfseries]

\newcommand{\msg}[6]{
      \ifthenelse{\isin{#1}{left} \AND \isin{#2}{down}}{
            \coordinate (anchor) at ($({#3})!{#5}!({#4})$);
            \node[xshift=-6.0mm] at (anchor) {#6};
            \node[xshift=-1.0mm] at (anchor) {$\downarrow$};
      }{}
      \ifthenelse{\isin{#1}{right} \AND \isin{#2}{down}}{
            \coordinate (anchor) at ($({#3})!{#5}!({#4})$);
            \node[xshift=6.0mm] at (anchor) {#6};
            \node[xshift=1.0mm] at (anchor) {$\downarrow$};
      }{}

      \ifthenelse{\isin{#1}{down} \AND \isin{#2}{right}}{
            \coordinate (anchor) at ($({#3})!{#5}!({#4})$);
            \node[ yshift=-4.0mm] at (anchor) {#6};
            \node[yshift=-1.0mm] at (anchor) {$\rightarrow$};
      }{}
      \ifthenelse{\isin{#1}{up} \AND \isin{#2}{right}}{
            \coordinate (anchor) at ($({#3})!{#5}!({#4})$);
            \node[ yshift=4.0mm] at (anchor) {#6};
            \node[yshift=1.0mm] at (anchor) {$\rightarrow$};
      }{}

      \ifthenelse{\isin{#1}{down} \AND \isin{#2}{left}}{
            \coordinate (anchor) at ($({#3})!{#5}!({#4})$);
            \node[ yshift=-4.0mm] at (anchor) {#6};
            \node[yshift=-1.0mm] at (anchor) {$\leftarrow$};
      }{}
      \ifthenelse{\isin{#1}{up} \AND \isin{#2}{left}}{
            \coordinate (anchor) at ($({#3})!{#5}!({#4})$);
            \node[ yshift=4.0mm] at (anchor) {#6};
            \node[yshift=1.0mm] at (anchor) {$\leftarrow$};
      }{}

      \ifthenelse{\isin{#1}{left} \AND \isin{#2}{up}}{
            \coordinate (anchor) at ($({#3})!{#5}!({#4})$);
            \node[ xshift=-6.0mm] at (anchor) {#6};
            \node[xshift=-1.0mm] at (anchor) {$\uparrow$};
      }{}
      \ifthenelse{\isin{#1}{right} \AND \isin{#2}{up}}{
            \coordinate (anchor) at ($({#3})!{#5}!({#4})$);
            \node[ xshift=6.0mm] at (anchor) {#6};
            \node[xshift=1.0mm] at (anchor) {$\uparrow$};
      }{}
}

\newcommand{\msgcircle}[6]{
      \ifthenelse{\isin{#1}{left} \AND \isin{#2}{down}}{
            \coordinate (anchor) at ($({#3})!{#5}!({#4})$);
            \node[msgcircle,xshift=-5.0mm] at (anchor) {#6};
            \node[xshift=-1.5mm] at (anchor) {$\downarrow$};
      }{}
      \ifthenelse{\isin{#1}{right} \AND \isin{#2}{down}}{
            \coordinate (anchor) at ($({#3})!{#5}!({#4})$);
            \node[msgcircle,xshift=5.0mm] at (anchor) {#6};
            \node[xshift=1.5mm] at (anchor) {$\downarrow$};
      }{}

      \ifthenelse{\isin{#1}{down} \AND \isin{#2}{right}}{
            \coordinate (anchor) at ($({#3})!{#5}!({#4})$);
            \node[msgcircle, yshift=-5.0mm] at (anchor) {#6};
            \node[yshift=-2.0mm] at (anchor) {$\rightarrow$};
      }{}
      \ifthenelse{\isin{#1}{up} \AND \isin{#2}{right}}{
            \coordinate (anchor) at ($({#3})!{#5}!({#4})$);
            \node[msgcircle, yshift=5.0mm] at (anchor) {#6};
            \node[yshift=2.0mm] at (anchor) {$\rightarrow$};
      }{}

      \ifthenelse{\isin{#1}{down} \AND \isin{#2}{left}}{
            \coordinate (anchor) at ($({#3})!{#5}!({#4})$);
            \node[msgcircle, yshift=-5.0mm] at (anchor) {#6};
            \node[yshift=-2.0mm] at (anchor) {$\leftarrow$};
      }{}
      \ifthenelse{\isin{#1}{up} \AND \isin{#2}{left}}{
            \coordinate (anchor) at ($({#3})!{#5}!({#4})$);
            \node[msgcircle, yshift=5.0mm] at (anchor) {#6};
            \node[yshift=2.0mm] at (anchor) {$\leftarrow$};
      }{}

      \ifthenelse{\isin{#1}{left} \AND \isin{#2}{up}}{
            \coordinate (anchor) at ($({#3})!{#5}!({#4})$);
            \node[msgcircle, xshift=-5.0mm] at (anchor) {#6};
            \node[xshift=-1.5mm] at (anchor) {$\uparrow$};
      }{}
      \ifthenelse{\isin{#1}{right} \AND \isin{#2}{up}}{
            \coordinate (anchor) at ($({#3})!{#5}!({#4})$);
            \node[msgcircle, xshift=5.0mm] at (anchor) {#6};
            \node[xshift=1.5mm] at (anchor) {$\uparrow$};
      }{}
}

\newcommand{\darkmsg}[6]{
      \ifthenelse{\isin{#1}{left} \AND \isin{#2}{down}}{
            \coordinate (anchor) at ($({#3})!{#5}!({#4})$);
            \node[darkmsgcircle, xshift=-5mm] at (anchor) {#6};
            \node[xshift=-1.5mm] at (anchor) {$\downarrow$};
      }{}
      \ifthenelse{\isin{#1}{right} \AND \isin{#2}{down}}{
            \coordinate (anchor) at ($({#3})!{#5}!({#4})$);
            \node[darkmsgcircle, xshift=5mm] at (anchor) {#6};
            \node[xshift=1.5mm] at (anchor) {$\downarrow$};
      }{}

      \ifthenelse{\isin{#1}{down} \AND \isin{#2}{right}}{
            \coordinate (anchor) at ($({#3})!{#5}!({#4})$);
            \node[darkmsgcircle, yshift=-5.0mm] at (anchor) {#6};
            \node[yshift=-2.0mm] at (anchor) {$\rightarrow$};
      }{}
      \ifthenelse{\isin{#1}{up} \AND \isin{#2}{right}}{
            \coordinate (anchor) at ($({#3})!{#5}!({#4})$);
            \node[darkmsgcircle, yshift=5.0mm] at (anchor) {#6};
            \node[yshift=2.0mm] at (anchor) {$\rightarrow$};
      }{}

      \ifthenelse{\isin{#1}{down} \AND \isin{#2}{left}}{
            \coordinate (anchor) at ($({#3})!{#5}!({#4})$);
            \node[darkmsgcircle, yshift=-5.0mm] at (anchor) {#6};
            \node[yshift=-2.0mm] at (anchor) {$\leftarrow$};
      }{}
      \ifthenelse{\isin{#1}{up} \AND \isin{#2}{left}}{
            \coordinate (anchor) at ($({#3})!{#5}!({#4})$);
            \node[darkmsgcircle, yshift=5.0mm] at (anchor) {#6};
            \node[yshift=2.0mm] at (anchor) {$\leftarrow$};
      }{}

      \ifthenelse{\isin{#1}{left} \AND \isin{#2}{up}}{
            \coordinate (anchor) at ($({#3})!{#5}!({#4})$);
            \node[darkmsgcircle, xshift=-5.0mm] at (anchor) {#6};
            \node[xshift=-1.5mm] at (anchor) {$\uparrow$};
      }{}
      \ifthenelse{\isin{#1}{right} \AND \isin{#2}{up}}{
            \coordinate (anchor) at ($({#3})!{#5}!({#4})$);
            \node[darkmsgcircle, xshift=5.0mm] at (anchor) {#6};
            \node[xshift=1.5mm] at (anchor) {$\uparrow$};
      }{}
}

\newcommand{\redbackmsg}[6]{
      \ifthenelse{\isin{#1}{left} \AND \isin{#2}{down}}{
            \coordinate (anchor) at ($({#3})!{#5}!({#4})$);
            \node[reddarkmsgcircle, xshift=-5mm] at (anchor) {#6};
            \node[xshift=-1.5mm] at (anchor) {$\downarrow$};
      }{}
      \ifthenelse{\isin{#1}{right} \AND \isin{#2}{down}}{
            \coordinate (anchor) at ($({#3})!{#5}!({#4})$);
            \node[reddarkmsgcircle, xshift=5mm] at (anchor) {#6};
            \node[xshift=1.5mm] at (anchor) {$\downarrow$};
      }{}

      \ifthenelse{\isin{#1}{down} \AND \isin{#2}{right}}{
            \coordinate (anchor) at ($({#3})!{#5}!({#4})$);
            \node[reddarkmsgcircle, yshift=-5.0mm] at (anchor) {#6};
            \node[yshift=-2.0mm] at (anchor) {$\rightarrow$};
      }{}
      \ifthenelse{\isin{#1}{up} \AND \isin{#2}{right}}{
            \coordinate (anchor) at ($({#3})!{#5}!({#4})$);
            \node[reddarkmsgcircle, yshift=5.0mm] at (anchor) {#6};
            \node[yshift=2.0mm] at (anchor) {$\rightarrow$};
      }{}

      \ifthenelse{\isin{#1}{down} \AND \isin{#2}{left}}{
            \coordinate (anchor) at ($({#3})!{#5}!({#4})$);
            \node[reddarkmsgcircle, yshift=-5.0mm] at (anchor) {#6};
            \node[yshift=-2.0mm] at (anchor) {$\leftarrow$};
      }{}
      \ifthenelse{\isin{#1}{up} \AND \isin{#2}{left}}{
            \coordinate (anchor) at ($({#3})!{#5}!({#4})$);
            \node[reddarkmsgcircle, yshift=5.0mm] at (anchor) {#6};
            \node[yshift=2.0mm] at (anchor) {$\leftarrow$};
      }{}

      \ifthenelse{\isin{#1}{left} \AND \isin{#2}{up}}{
            \coordinate (anchor) at ($({#3})!{#5}!({#4})$);
            \node[reddarkmsgcircle, xshift=-5.0mm] at (anchor) {#6};
            \node[xshift=-1.5mm] at (anchor) {$\uparrow$};
      }{}
      \ifthenelse{\isin{#1}{right} \AND \isin{#2}{up}}{
            \coordinate (anchor) at ($({#3})!{#5}!({#4})$);
            \node[reddarkmsgcircle, xshift=5.0mm] at (anchor) {#6};
            \node[xshift=1.5mm] at (anchor) {$\uparrow$};
      }{}
}

\newcommand{\redmsg}[6]{
      \ifthenelse{\isin{#1}{left} \AND \isin{#2}{down}}{
            \coordinate (anchor) at ($({#3})!{#5}!({#4})$);
            \node[redmsgcircle, xshift=-5mm] at (anchor) {#6};
            \node[xshift=-1.5mm] at (anchor) {$\downarrow$};
      }{}
      \ifthenelse{\isin{#1}{right} \AND \isin{#2}{down}}{
            \coordinate (anchor) at ($({#3})!{#5}!({#4})$);
            \node[redmsgcircle, xshift=5mm] at (anchor) {#6};
            \node[xshift=1.5mm] at (anchor) {$\downarrow$};
      }{}

      \ifthenelse{\isin{#1}{down} \AND \isin{#2}{right}}{
            \coordinate (anchor) at ($({#3})!{#5}!({#4})$);
            \node[redmsgcircle, yshift=-5.0mm] at (anchor) {#6};
            \node[yshift=-2.0mm] at (anchor) {$\rightarrow$};
      }{}
      \ifthenelse{\isin{#1}{up} \AND \isin{#2}{right}}{
            \coordinate (anchor) at ($({#3})!{#5}!({#4})$);
            \node[redmsgcircle, yshift=5.0mm] at (anchor) {#6};
            \node[yshift=2.0mm] at (anchor) {$\rightarrow$};
      }{}

      \ifthenelse{\isin{#1}{down} \AND \isin{#2}{left}}{
            \coordinate (anchor) at ($({#3})!{#5}!({#4})$);
            \node[redmsgcircle, yshift=-5.0mm] at (anchor) {#6};
            \node[yshift=-2.0mm] at (anchor) {$\leftarrow$};
      }{}
      \ifthenelse{\isin{#1}{up} \AND \isin{#2}{left}}{
            \coordinate (anchor) at ($({#3})!{#5}!({#4})$);
            \node[redmsgcircle, yshift=5.0mm] at (anchor) {#6};
            \node[yshift=2.0mm] at (anchor) {$\leftarrow$};
      }{}

      \ifthenelse{\isin{#1}{left} \AND \isin{#2}{up}}{
            \coordinate (anchor) at ($({#3})!{#5}!({#4})$);
            \node[redmsgcircle, xshift=-5.0mm] at (anchor) {#6};
            \node[xshift=-1.5mm] at (anchor) {$\uparrow$};
      }{}
      \ifthenelse{\isin{#1}{right} \AND \isin{#2}{up}}{
            \coordinate (anchor) at ($({#3})!{#5}!({#4})$);
            \node[redmsgcircle, xshift=5.0mm] at (anchor) {#6};
            \node[xshift=1.5mm] at (anchor) {$\uparrow$};
      }{}
}

\newcommand{\bwmsg}[6]{
      \ifthenelse{\isin{#1}{left} \AND \isin{#2}{down}}{
            \coordinate (anchor) at ($({#3})!{#5}!({#4})$);
            \node[msgdoublecircle, xshift=-5.5mm] at (anchor) {#6};
            \node[xshift=-1.5mm] at (anchor) {$\downarrow$};
      }{}
      \ifthenelse{\isin{#1}{right} \AND \isin{#2}{down}}{
            \coordinate (anchor) at ($({#3})!{#5}!({#4})$);
            \node[msgdoublecircle, xshift=5.5mm] at (anchor) {#6};
            \node[xshift=1.5mm] at (anchor) {$\downarrow$};
      }{}

      \ifthenelse{\isin{#1}{down} \AND \isin{#2}{right}}{
            \coordinate (anchor) at ($({#3})!{#5}!({#4})$);
            \node[msgdoublecircle, yshift=-6.0mm] at (anchor) {#6};
            \node[yshift=-2.0mm] at (anchor) {$\rightarrow$};
      }{}
      \ifthenelse{\isin{#1}{up} \AND \isin{#2}{right}}{
            \coordinate (anchor) at ($({#3})!{#5}!({#4})$);
            \node[msgdoublecircle, yshift=6.0mm] at (anchor) {#6};
            \node[yshift=2.0mm] at (anchor) {$\rightarrow$};
      }{}

      \ifthenelse{\isin{#1}{down} \AND \isin{#2}{left}}{
            \coordinate (anchor) at ($({#3})!{#5}!({#4})$);
            \node[msgdoublecircle, yshift=-6.0mm] at (anchor) {#6};
            \node[yshift=-2.0mm] at (anchor) {$\leftarrow$};
      }{}
      \ifthenelse{\isin{#1}{up} \AND \isin{#2}{left}}{
            \coordinate (anchor) at ($({#3})!{#5}!({#4})$);
            \node[msgdoublecircle, yshift=6.0mm] at (anchor) {#6};
            \node[yshift=2.0mm] at (anchor) {$\leftarrow$};
      }{}

      \ifthenelse{\isin{#1}{left} \AND \isin{#2}{up}}{
            \coordinate (anchor) at ($({#3})!{#5}!({#4})$);
            \node[msgdoublecircle, xshift=-5.5mm] at (anchor) {#6};
            \node[xshift=-1.5mm] at (anchor) {$\uparrow$};
      }{}
      \ifthenelse{\isin{#1}{right} \AND \isin{#2}{up}}{
            \coordinate (anchor) at ($({#3})!{#5}!({#4})$);
            \node[msgdoublecircle, xshift=5.5mm] at (anchor) {#6};
            \node[xshift=1.5mm] at (anchor) {$\uparrow$};
      }{}
}

\newcommand{\bwdarkmsg}[6]{
      \ifthenelse{\isin{#1}{left} \AND \isin{#2}{down}}{
            \coordinate (anchor) at ($({#3})!{#5}!({#4})$);
            \node[darkmsgdoublecircle, xshift=-5.5mm] at (anchor) {#6};
            \node[xshift=-1.5mm] at (anchor) {$\downarrow$};
      }{}
      \ifthenelse{\isin{#1}{right} \AND \isin{#2}{down}}{
            \coordinate (anchor) at ($({#3})!{#5}!({#4})$);
            \node[darkmsgdoublecircle, xshift=5.5mm] at (anchor) {#6};
            \node[xshift=1.5mm] at (anchor) {$\downarrow$};
      }{}

      \ifthenelse{\isin{#1}{down} \AND \isin{#2}{right}}{
            \coordinate (anchor) at ($({#3})!{#5}!({#4})$);
            \node[darkmsgdoublecircle, yshift=-6.0mm] at (anchor) {#6};
            \node[yshift=-2.0mm] at (anchor) {$\rightarrow$};
      }{}
      \ifthenelse{\isin{#1}{up} \AND \isin{#2}{right}}{
            \coordinate (anchor) at ($({#3})!{#5}!({#4})$);
            \node[darkmsgdoublecircle, yshift=6.0mm] at (anchor) {#6};
            \node[yshift=2.0mm] at (anchor) {$\rightarrow$};
      }{}

      \ifthenelse{\isin{#1}{down} \AND \isin{#2}{left}}{
            \coordinate (anchor) at ($({#3})!{#5}!({#4})$);
            \node[darkmsgdoublecircle, yshift=-6.0mm] at (anchor) {#6};
            \node[yshift=-2.0mm] at (anchor) {$\leftarrow$};
      }{}
      \ifthenelse{\isin{#1}{up} \AND \isin{#2}{left}}{
            \coordinate (anchor) at ($({#3})!{#5}!({#4})$);
            \node[darkmsgdoublecircle, yshift=6.0mm] at (anchor) {#6};
            \node[yshift=2.0mm] at (anchor) {$\leftarrow$};
      }{}

      \ifthenelse{\isin{#1}{left} \AND \isin{#2}{up}}{
            \coordinate (anchor) at ($({#3})!{#5}!({#4})$);
            \node[darkmsgdoublecircle, xshift=-5.5mm] at (anchor) {#6};
            \node[xshift=-1.5mm] at (anchor) {$\uparrow$};
      }{}
      \ifthenelse{\isin{#1}{right} \AND \isin{#2}{up}}{
            \coordinate (anchor) at ($({#3})!{#5}!({#4})$);
            \node[darkmsgdoublecircle, xshift=5.5mm] at (anchor) {#6};
            \node[xshift=1.5mm] at (anchor) {$\uparrow$};
      }{}
}

\tikzset{mainstyle/.style={fill=white, draw=black, shape=rectangle, align=center}}

\tikzset{dstyle/.style={mainstyle, minimum size=4mm, inner sep=0pt, text width=4mm}}

\tikzset{sstyle/.style={mainstyle, minimum size=5mm, inner sep=0pt, text width=5mm}}

\tikzset{ostyle/.style={fill=darkgrey, draw=black, shape=rectangle, minimum size=0.2cm, inner sep=0pt, text width=2mm}}

\tikzstyle{observation}=[ostyle]
\tikzstyle{deterministic}=[dstyle]
\tikzstyle{stochastic}=[sstyle]

\tikzstyle{filter}=[mainstyle, minimum width=1cm, minimum height=0.5cm]
\tikzstyle{selector}=[fill=white, draw=black, shape=trapezium, rotate=180, minimum width=1cm, minimum height=0.5cm]

\def\sfy{\mathsf{y}}
\def\sfx{\mathsf{x}}
\def\sfu{\mathsf{u}}
\def\sfg{\mathsf{g}}
\def\rmd{\mathrm{d}}

\def\-{\text{-}}
\def\+{\text{+}}
\def\tm{\! - \!}

\newcommand{\refappx}[1]{\hyperref[#1]{Appendix~\ref*{#1}}}
\newcommand\given[1][]{\:#1\vert\:}

\newcommand{\E}{\mathbb{E}}

\begin{document}
\title{Expected free energy as an information constraint on the Bethe Lagrangian}
%
%
\author{Wouter M. Kouw}
\authorrunning{W.M. Kouw}
%
\institute{TU Eindhoven, Eindhoven, Netherlands\\
\email{w.m.kouw@tue.nl}
}
\maketitle              
\begin{abstract}
Active inference selects actions by minimising an \emph{expected} free energy functional over predicted futures. However, adding an expectation over yet-unobserved outcomes means the free energy functional no longer has a Kullback--Leibler structure, which hinders message passing treatments of inference procedures. We propose an alternative formulation based on a Bethe free energy functional, fully supporting inference by message passing. The epistemic drive is maintained by imposing an information constraint, next to normalisation, marginalisation and form constraints, insisting that the mutual information between future observations, states and parameters given actions must be at least as large as the entropy of the goal prior. For a specific value of the corresponding Karush--Kuhn--Tucker multiplier, the stationary point of this constrained Bethe Lagrangian recovers the expected free energy solution. We show that, as the information demand is varied, the solved multiplier moves through its inactive, interior, and saturated regimes. In the inactive regime the agent's epistemic drive switches off entirely, while in the saturated regime it is maximal. We compare the performance of the constrained Bethe agent on three tasks against EFE and Q-MDP.
\end{abstract}
\section{Introduction}

Active inference frames adaptive decision-making under uncertainty as the minimisation of an expected free energy functional (EFE) that, unlike the objectives of reinforcement learning, naturally balances exploration and exploitation~\cite{da2020active,parr2022active,ladosz2022exploration}.
Scaling active inference in structured generative models is typically done by casting inference as message passing on factor graphs~\cite{loeliger2007factor,de2017factor,champion2021realizing,friston2025pixels,nuijten2025message,kouw2025message}. The analogue of Friston's free energy functional on factor graphs is the Bethe free energy, which distributes energy terms to nodes, entropy terms to edges and avoids over-counting entropies around deterministic operation nodes \cite{pearl1988probabilistic,yedidia2005constructing,csenoz2021variational}. The stationary points of the Bethe free energy functional correspond to messages that are passed between nodes and multiplied at edges to form marginal distributions. Sum-product and variational message passing as well as expectation-propagation all arise as stationary points but under different constraints on the functional~\cite{zhang2021unifying,csenoz2021variational}. 
The Bethe free energy has the structure of a Kullback--Leibler (KL) divergence, which allows its terms to be cleanly distributed over nodes and edges. EFE, however, does not have the structure of a KL divergence due to its expectation over yet-unobserved future outcomes. As such, its stationary points cannot be immediately distributed over the model's corresponding factor graph~\cite{van2022active,van2024realizing,de2025expected,nuijten2025message}. 
%
Efforts to address this discrepancy have been diverse. Schwöbel et al. propagated beliefs over future trajectories under a fixed expected-free-energy policy score \cite{schwobel2018active}, essentially leaving EFE out of the message passing procedure. Van de Laar et al. introduced a point-mass constraint on predicted outcomes that ensures that the constrained Bethe free energy functional experiences an epistemic drive \cite{van2022active}. Recently, Kouw et al. accepted the mismatch and simply passed the EFE solution as a non-standard message on the graph \cite{kouw2025message}, while Nuijten et al. added epistemic priors to the generative model such that minimisation of the Bethe free energy coincides with EFE minimisation \cite{nuijten2025message}.

We follow Van de Laar et al.'s direction in that we formulate active inference as minimisation of a constrained Bethe free energy functional, but we impose an information constraint instead of a point-mass. Specifically, we require that the mutual information between observations, states and parameters given actions, averaged over the variational posterior over actions, must be at least as large as the entropy of the goal prior distribution. Minimising the constrained Bethe Lagrangian under normalisation, marginalisation, form, and information constraints yields a stationary policy (see Eq.~\ref{eq:policy-star}) that combines goal-seeking and information-seeking terms. The information-seeking term has a weight $\gamma$ (Section~\ref{sec:stationary_points}) that arises from the Karush--Kuhn--Tucker (KKT) multiplier on the information constraint and is solved rather than tuned \cite{boyd2004convex}. At $\gamma = 1$, EFE is recovered (see \autoref{prop:efe}). At $\gamma = 0$, the constraint is inactive and the agent experiences no epistemic drive. We test this constrained Bethe agent on three tasks, comparing it to standard EFE minimisation and to Q-MDP \cite{littman1995learning}.

\section{Problem statement}

Consider an agent that interacts with its environment in a discrete-time manner. It makes actions $u_t \in \mathcal{U}$ and observes outcomes $y_t \in \mathcal{Y}$, where $\mathcal{U}$ and $\mathcal{Y}$ are finite sets. It employs a probabilistic state-space model with discrete states $x_t \in \mathcal{X}$, transition parameters $B$, and emission parameters $A$.

At time step $t$, to plan for its next action, the agent unrolls the model forwards in time to a horizon of $K$ time steps, indexed by $k$ and collected as $\tau = (t+1, t+2, \dots, t+K)$. After attaching a goal prior $p_*(y_k)$ to every predicted observation, the generative model becomes
\begin{align}\label{eq:generative-planning-model}
    p(&\,y_{\tau}, u_\tau, x_{\tau}, x_t, B, A \given \mathcal{D}_t) \propto  \\
    &\underbrace{p(x_t \given \mathcal{D}_t)}_{f_t}  \underbrace{p(B \given \mathcal{D}_t)}_{f_B} \underbrace{p(A \given \mathcal{D}_t)}_{f_A} \prod_{k \in \tau} \underbrace{p(x_k \given x_{k-1}, u_k, B)}_{f_x}  \underbrace{p(y_k \given x_k, A)}_{f_y} \underbrace{p(u_k)}_{f_u} \underbrace{p_*(y_k)}_{f_*} , \nonumber
\end{align}
where $\mathcal{D}_t = \{y_{1:t}, u_{1:t}\}$. Note that $t$ indexes trial time and $k$ indexes planning time steps, and that the proportionality is due to the multiplication with the goal priors, leaving the product unnormalised. For brevity we collect the model parameters as $\Theta = \{A, B\}$ and the policy as $u_\tau = (u_{t+1}, \dots, u_{t+K})$ with prior $p(u_\tau) = \prod_{k\in\tau} p(u_k)$. The Forney-style factor graph of~\eqref{eq:generative-planning-model} is shown in \autoref{fig:ffg-planning}.

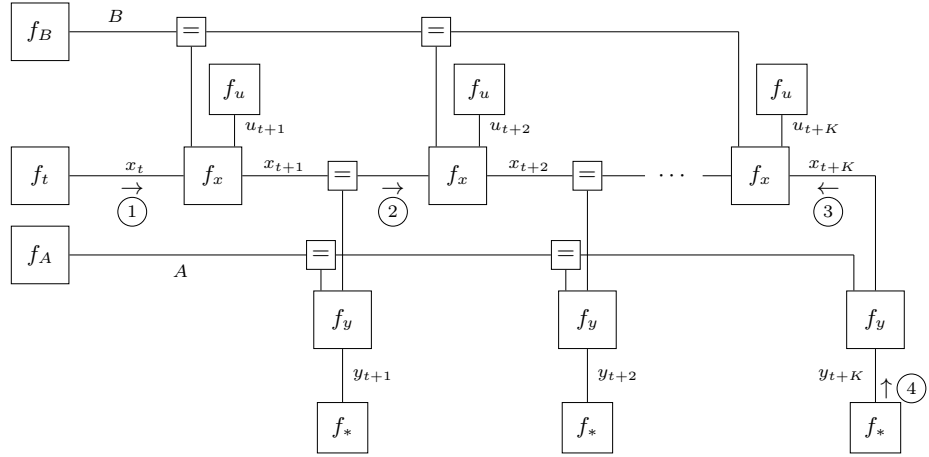
\begin{figure}[htb]
    \resizebox{\textwidth}{!}{\begin{tikzpicture}[every node/.style={font=\small}]


    \node [style=stochastic, minimum width=8mm, minimum height=8mm]   (pxk) at (-4., 0) {$f_t$};
    \node [draw, rectangle, minimum width=8mm, minimum height=8mm]     (f1)  at (-1.6, 0) {$f_x$};
    \node [style=deterministic]                                        (ex1) at ( 0.2, 0) {$=$};
    \node [draw, rectangle, minimum width=8mm, minimum height=8mm]     (f2)  at ( 1.8, 0) {$f_x$};
    \node [style=deterministic]                                        (ex2) at ( 3.6, 0) {$=$};
    \node [draw, rectangle, minimum width=8mm, minimum height=8mm]     (f3)  at ( 6.0, 0) {$f_x$};
    \coordinate                                                         (xT)  at ( 7.6, 0);

    \node [draw, rectangle, minimum width=8mm, minimum height=8mm] (g1) at (0.2, -2.0) {$f_y$};
    \node [draw, rectangle, minimum width=8mm, minimum height=8mm] (g2) at (3.6, -2.0) {$f_y$};
    \node [draw, rectangle, minimum width=8mm, minimum height=8mm] (g3) at (7.6, -2.0) {$f_y$};
    
    \node [style=stochastic, minimum width=7mm, minimum height=7mm] (y1) at (0.2, -3.5) {$f_*$};
    \node [style=stochastic, minimum width=7mm, minimum height=7mm] (y2) at (3.6, -3.5) {$f_*$};
    \node [style=stochastic, minimum width=7mm, minimum height=7mm] (y3) at (7.6, -3.5) {$f_*$};
    \node [style=edgelabel, right=1.mm] at ($(g1.south)!0.5!(y1.north)$) {$y_{t+1}$};
    \node [style=edgelabel, right=1.mm] at ($(g2.south)!0.5!(y2.north)$) {$y_{t+2}$};
    \node [style=edgelabel, left=1.mm] at ($(g3.south)!0.5!(y3.north)$) {$y_{t+K}$};

    \node [style=stochastic, minimum width=7mm, minimum height=7mm] (u1) at (-1.3, 1.2) {$f_u$};
    \node [style=stochastic, minimum width=7mm, minimum height=7mm] (u2) at ( 2.1, 1.2) {$f_u$};
    \node [style=stochastic, minimum width=7mm, minimum height=7mm] (u3) at ( 6.3, 1.2) {$f_u$};
    \node [right=0.5mm of u1, font=\scriptsize] {};
    \node [right=0.5mm of u2, font=\scriptsize] {};
    \node [right=0.5mm of u3, font=\scriptsize] {};

    \node [style=stochastic, minimum width=8mm, minimum height=8mm]    (pth) at (-4, 2) {$f_B$};
    \node [style=deterministic] (th1) at (-1.9, 2) {$=$};
    \node [style=deterministic] (th2) at ( 1.5, 2) {$=$};
    \node [] (th3) at ( 5.7, 2) {};

    \node [style=stochastic, minimum width=8mm, minimum height=8mm]    (pze) at (-4, -1.1) {$f_A$};
    \node [style=deterministic] (ze1) at (-0.1, -1.1) {$=$};
    \node [style=deterministic] (ze2) at ( 3.3, -1.1) {$=$};
    \node [] (ze3) at ( 7.3, -1.1) {};


    \draw[-] (pxk) -- (f1);
    \draw[-] (f1)  -- (ex1);
    \draw[-] (ex1) -- (f2);
    \draw[-] (f2)  -- (ex2);
    \draw[-] (ex2) -- (4.4, 0);
    \node at (4.8, 0) {$\cdots$};
    \draw[-] (5.2, 0) -- (f3);
    \draw[-] (f3)  -- (xT);

    \node [style=edgelabel, above] at ($(pxk)!0.55!(f1)$)   {$x_t$};
    \node [style=edgelabel, above] at ($(f1)!0.55!(ex1)$) {$x_{t+1}$};
    \node [style=edgelabel, above] at ($(f2)!0.55!(ex2)$) {$x_{t+2}$};
    \node [style=edgelabel, above] at ($(f3.east)!0.5!(xT)$) {$x_{t+K}$};

    \draw[-] (ex1) -- (g1);
    \draw[-] (ex2) -- (g2);
    \draw[-] (xT)  -- (g3);
    \draw[-] (g1) -- (y1);
    \draw[-] (g2) -- (y2);
    \draw[-] (g3) -- (y3);

    \draw[-] (u1) -- ([xshift=3mm]f1.north);
    \draw[-] (u2) -- ([xshift=3mm]f2.north);
    \draw[-] (u3) -- ([xshift=3mm]f3.north);
    \node [style=edgelabel, right = 1.mm] at ($(u1)!0.5!([xshift=3mm, yshift=-3mm]f1.north)$) {$u_{t+1}$};
    \node [style=edgelabel, right = 1.mm] at ($(u2)!0.5!([xshift=3mm, yshift=-3mm]f2.north)$) {$u_{t+2}$};
    \node [style=edgelabel, right = 1.mm] at ($(u3)!0.5!([xshift=3mm, yshift=-3mm]f3.north)$) {$u_{t+K}$};

    \draw[-] (pth) -- (th1) -- (th2) -- (th3.center);
    \draw[-] (th1) -- ([xshift=-3mm]f1.north);
    \draw[-] (th2) -- ([xshift=-3mm]f2.north);
    \draw[-] (th3.center) -- ([xshift=-3mm]f3.north);
    \node [style=edgelabel, above = 1.mm] at ($(pth)!0.5!(th1)$) {$B$};

    \draw[-] (pze) -- (ze1) -- (ze2) -- (ze3.center);
    \draw[-] (ze1) -- ([xshift=-3mm]g1.north);
    \draw[-] (ze2) -- ([xshift=-3mm]g2.north);
    \draw[-] (ze3.center) -- ([xshift=-3mm]g3.north);
    \node [style=edgelabel, below = 1.mm] at ($(pze)!0.5!(ze1)$) {$A$};

    \msgcircle{down}{right}{pxk.east}{f1.west}{0.55}{1}
    \msgcircle{down}{right}{ex1.east}{f2.west}{0.5}{2}
    \msgcircle{down}{left}{xT}{f3.east}{0.55}{3}
    \msgcircle{right}{up}{y3}{g3.south}{0.5}{4}

\end{tikzpicture}}
    \vspace{-5pt}
    \caption{Forney-style factor graph of the planning model in~\eqref{eq:generative-planning-model} over the horizon $\tau$. Each transition factor ($f_x$) maps the previous state $x_{k-1}$, the control $u_k$ and the transition parameters $B$ to $x_k$. Each emission factor ($f_y$) maps $x_k$ and the emission parameters $A$ to the predicted observation $y_k$, whose edge terminates in the goal-prior factor $f_* = p_*(y_k)$. Equality nodes ($=$) share $B$ and $A$ across time. Forward predictive messages (1, 2) and backward messages (3, 4) realise the forward--backward sweep of Section~\ref{sec:implementation}.}
    \label{fig:ffg-planning}
    \vspace{-10pt}
\end{figure}

\subsubsection{Expected free energy} Expected free energy scores candidate actions by the surprise an agent expects to incur about its future observations under its preferences~\cite{da2020active}. Writing $q(y_\tau, x_\tau, u_\tau, A, B) = q(y_\tau \given x_\tau, A)\, q(x_\tau, u_\tau, A, B)$ for the variational model, the expected free energy functional is
\begin{equation} \label{eq:efe-functional}
	\mathcal{G}[q] = \mathbb{E}_{q(x_\tau, u_\tau, A, B)} \Big[ \mathbb{E}_{q(y_\tau \given x_\tau, A)} \big[ \ln \frac{q(x_\tau, u_\tau, A, B)}{p(y_{\tau}, u_{\tau}, x_{\tau},  x_t, B, A \given \mathcal{D}_t)} \big] \Big] \, .
\end{equation}
For planning models with independent control priors $p(u_k)$, the optimal marginal variational factor for controls is $q_{\mathrm{EFE}}(u_k) \propto p(u_k) \exp( - G(u_k))$ \cite{van2022active}. The expected free energy function $G(u_k)$ decomposes into risk and ambiguity terms:
\begin{align}\label{eq:efe-def}
    G(u_k) &= \E_{q(y_k, x_k, B, A \given u_k)}\!\left[\ln \frac{q(x_k, B, A \given u_k)}{p(y_{k}, x_{k}, B, A \given u_k, \mathcal{D}_t)}\right] \\
    &= \underbrace{D_{\mathrm{KL}}\big[q(y_k \given u_k) \,\|\, p_*(y_k) \big]}_{\text{risk}} + \underbrace{\E_{q(x_k, B, A | u_k)}\Big[H\big[p(y_k \given x_k, B, A, u_k) \big] \Big]}_{\text{ambiguity}}  , \label{eq:efe-risk-ambiguity}
\end{align}
where $D_{\mathrm{KL}}[\cdot]$ denotes KL-divergence and $H[\cdot]$ denotes entropy \cite{cover1999elements}.
Risk penalises divergence between predicted outcomes and the goal prior, while ambiguity penalises outcomes that are uninformative about the latent states and parameters. Thus, exploitation (low risk) and exploration (high information gain) emerge from a single objective \cite{da2020active}. Re-arranging terms exposes the information gain:
\begin{equation}
I\big[x_k, \Theta;\, y_k \given u_k \big] = H\big[q(y_k \given u_k) \big] - \E_{q(x_k,\Theta \given u_k)} \Big[H \big[p(y_k \given x_k, \Theta) \big] \Big] \, .
\end{equation} 
This is the predictive mutual information between the latents and the outcome, which allows us to express the negative expected free energy as a goal cross-entropy plus an information gain
\begin{equation}\label{eq:efe-identity}
    -G(u_k) \;=\; \E_{q(y_k \given u_k)} \big[\ln p_*(y_k)\big] \;+\; I\big[x_k, \Theta \, ;\, y_k \given u_k\big]\, .
\end{equation}
However, the expectation over unobserved outcomes in~\eqref{eq:efe-functional}, i.e., $\mathbb{E}_{q(y_\tau | x_\tau, A)}[\cdot]$, means the functional is not a Kullback--Leibler divergence. As such, it does not decompose into one average energy per factor and one entropy per edge, and its stationary points therefore do not correspond to the variational message updates~\cite{van2024realizing,de2025expected,nuijten2025message}. Restoring that structure does not help; incorporating $q(y_\tau \given x_\tau, A)$ into the numerator of~\eqref{eq:efe-functional} yields a proper variational free energy, but the emission entropy it contributes cancels the emission energy exactly (see \autoref{lem:collapse}). So the ambiguity term vanishes with it and only risk survives. 

The next section demonstrates how one may obtain the EFE solution (Eq.~\ref{eq:efe-identity}) by minimising the Bethe free energy functional under the right set of constraints.

\section{Constrained Bethe free energy}

The Bethe free energy of the planning graph collects one average-energy term per factor of~\eqref{eq:generative-planning-model}, one entropy term for every variable and a term that prevents over-counting entropies based on the degree of each variable,
\begin{equation}\label{eq:bethe-def}
\begin{aligned}
    \mathcal{B}[q]
    =&\, \E_{q}\!\left[\ln \frac{q(x_t)}{p(x_t\given\mathcal{D}_t)}\right]  + \E_{q}\!\left[\ln \frac{q(B)}{p(B\given\mathcal{D}_t)}\right] + \E_{q}\!\left[\ln \frac{q(A)}{p(A\given\mathcal{D}_t)}\right]\\
    &+ \sum_{k\in\tau} \E_{q}\!\left[\ln \frac{q(x_k, x_{k-1}, u_k, B)}{p(x_k \given x_{k-1}, u_k, B)}\right]
      + \sum_{k\in\tau} \E_{q}\!\left[\ln \frac{q(y_k, x_k, A)}{p(y_k \given x_k, A)}\right] \\
    &+ \sum_{k\in\tau} \E_{q}\!\left[\ln \frac{q(u_k)}{p(u_k)}\right]
      + \sum_{k\in\tau} \E_{q}\!\left[\ln \frac{q(y_k)}{p_*(y_k)}\right]
      \,-\sum_{s\in\mathcal{S}} (d_s \tm 1)\,\E_{q}[\ln q(s)]\,,
\end{aligned}
\end{equation}
where $\mathcal{S} = \{A, B, x_t \} \cup \{u_k, x_k, y_k\}_{k\in\tau}$ collects the variables shared across factors.
The degrees are $d_B = d_A = K+1$ (one parameter prior plus $K$ transitions or emissions), $d_{x_t}=d_{x_{t+K}}=2$, $d_{x_k}=3$ for $k\in\tau\setminus\{t+K\}$, and $d_{u_k}=2$.
The goal-prior factors raise the degree of the predicted observations to $d_{y_k}=2$. 

\subsubsection{Constraints}

We impose four families of constraints on the variational factors. Firstly, normalisation, requiring all variational factors to be valid probability distributions, and secondly, consistency across marginalisation for every cluster--edge incidence $(a,i)$ in the planning graph: 
\begin{equation}
	\sum_{s_i} q_i(s_i) = 1 \, , \qquad \qquad \sum_{s_a\setminus s_i} q_a(s_a) = q_i(s_i) \, .
\end{equation}
Thirdly, we impose form constraints that pin the rolled out beliefs to the model's own forward predictions.

\begin{definition}[Form constraints]\label{def:form-constraint}
  The variational factors of the current state and parameter beliefs are constrained to match the generative model, 
    \begin{equation}
    q(x_t) = p(x_t \given \mathcal{D}_t), \qquad q(B) = p(B\given\mathcal{D}_t) \quad \text{and} \quad q(A) = p(A\given\mathcal{D}_t) \, .
    \end{equation}
    For $k \in \tau$, the variational factors of the state transition and likelihood match the generative model,
    \begin{equation}
    q(x_k \given x_{k-1}, u_k, B) = p(x_k \given x_{k-1}, u_k, B) \quad \text{and} \quad q(y_k \given x_k, A) = p(y_k \given x_k, A) \, .
    \end{equation}
\end{definition}
In short, we do not infer the state prior, the parameter beliefs, the structure of the state transition and the structure of the emission likelihood during action planning.
Lastly, we impose a constraint on the amount of information we expect to gain under our action posterior.
\begin{definition}[Information constraint]
For each $k\in\tau$, the expected joint mutual information between the latent state and parameters on the one hand, and the predicted observation on the other, is at least the goal-prior entropy:
    \begin{equation}\label{eq:info-constraint}
        \E_{q(u_\tau)}\Big[I\big[x_k,\Theta;\, y_k \mid u_\tau\big] \Big] \;\ge\; H\big[p_*(y_k) \big] \, .
    \end{equation}
\end{definition}
The information $I\big[x_k,\Theta;\,y_k\given u_\tau \big]$ is evaluated under the form constraints, so it is a fixed function of $u_\tau$.
It decomposes into state information gain $I\big[x_k;\,y_k\mid\Theta,u_\tau \big]$ (salience), and parameter information gain $I\big[\Theta;\,y_k\mid u_\tau \big]$ (novelty), thus addressing both epistemic drives.

The lower bound, i.e., the information floor $\beta_k \equiv H\big[p_*(y_k)\big]$, is based on the goal prior rather than set by hand or tuned as a hyperparameter. The agent must make the predicted outcome as informative about the latents as its preferences require, which is also what EFE incentivises (see \autoref{prop:efe}).

\subsubsection{Stationary points} \label{sec:stationary_points}
The form constraints ensure the only free belief left on the planning graph is the policy belief $q(u_\tau)$, drastically simplifying the Bethe free energy functional.
\begin{lemma}[Bethe collapse]\label{lem:collapse}
Under the form constraints of \autoref{def:form-constraint}, the Bethe free energy~\eqref{eq:bethe-def} of the goal-augmented planning graph reduces exactly to
\begin{equation}\label{eq:bethe-collapse}
    \mathcal{B}[q] \;=\; D_{\mathrm{KL}}\big[q(u_\tau)\,\big\|\,p(u_\tau)\big] \;-\; \sum_{k\in\tau} \E_{q(u_\tau)}\, \Big[ \E_{q(y_k\given u_\tau)}\big[\ln p_*(y_k)\big] \Big] \,.
\end{equation}
\end{lemma}

\begin{proof}[sketch]
Given $u_\tau$, the constrained family equals the model's own factorisation, so every transition and likelihood energy cancels against the corresponding entropy term. Only the policy prior and the goal factors survive.  \qed
\end{proof}
The full accounting is given in Appendix \ref{app:proof_lem1}.

Introducing a multiplier $\lambda$ for the normalisation of $q(u_\tau)$ and $\gamma_k \ge 0$ for the information constraint at each $k\in\tau$, the constrained Lagrangian reads
\begin{equation}\label{eq:lagrangian}
    \mathcal{L}[q] = \mathcal{B}[q]
    + \lambda\Big(\sum_{u_\tau} q(u_\tau) - 1\Big)
    + \sum_{k\in\tau} \gamma_k \Big(\beta_k - \E_{q(u_\tau)}\Big[ I \big[x_k,\Theta;\, y_k \given u_\tau \big] \Big] \Big)\,,
\end{equation}
where the information term is written in Karush--Kuhn--Tucker form  \cite{boyd2004convex}. The multiplier $\gamma_k$ must be non-negative. Wherever the bound is slack, $\gamma_k = 0$. So, the multiplier switches on only when the agent would otherwise be insufficiently informative.

\begin{theorem}[Stationary policy]\label{thm:policy}
Every stationary point of the Lagrangian~\eqref{eq:lagrangian} subject to the form constraints (\autoref{def:form-constraint}) has policy belief
\begin{equation}\label{eq:policy-star}
    q^\star(u_\tau) \;\propto\; p(u_\tau)\, \exp\Big( \sum_{k\in\tau} \Big( \E_{q(y_k \given u_\tau)}\big[\ln p_*(y_k)\big] + \gamma_k\, I\big[x_k,\Theta\, ;\, y_k \mid u_\tau\big] \Big) \Big)\,,
\end{equation}
with $\gamma_k \ge 0$, and control marginals $q^\star(u_k) = \sum_{u_\tau \setminus u_k} q^\star(u_\tau)$.
\end{theorem}

\begin{proof}[sketch]
By \autoref{lem:collapse} and the linearity of the constraint terms in $q(u_\tau)$, the variation of~\eqref{eq:lagrangian} with respect to $q(u_\tau)$ is 
\begin{equation}
\begin{split}
\frac{\delta\mathcal{L}}{\delta q(u_\tau)} =&  \ln q(u_\tau) + 1 - \ln p(u_\tau) + \lambda \\
&-  \sum_{k\in\tau} \Big(\E_{q(y_k | u_\tau)}\big[ \ln p_*(y_k)\big]  + \gamma_k I\big[x_k,\Theta;\,y_k | u_\tau\big]\Big) .
\end{split}
\end{equation}
Setting it to zero and normalising yields~\eqref{eq:policy-star}. $\gamma_k \ge 0$ holds by the Karush--Kuhn--Tucker conditions for~\eqref{eq:info-constraint}. \qed
\end{proof}
The full proof is in Appendix \ref{app:proof_thm1}. The multipliers $\{\gamma_k\}$ are determined implicitly by the information constraints~\eqref{eq:info-constraint} and complementary slackness (\autoref{lem:gamma} in \refappx{appx:derivation}), while $\lambda$ is fixed by normalisation.
The stationary policy~\eqref{eq:policy-star} modulates the action prior by an exponential of the goal cross-entropy and the joint state-parameter--observation mutual information, whose weight $\gamma_k$ is solved rather than tuned. The relation to expected free energy is exact:

\begin{proposition}[Equivalence to expected free energy]\label{prop:efe}
If $\gamma_k = 1$ for all $k\in\tau$, then the stationary policy~\eqref{eq:policy-star} coincides exactly with the expected-free-energy policy,
\begin{equation}\label{eq:efe-equivalence}
    q^\star(u_\tau) \, = \, q_{\mathrm{EFE}}(u_\tau)  \propto \, p(u_\tau) \exp\big(-G(u_\tau)\big)\, ,
\end{equation}
where $G(u_\tau) = \sum_{k\in\tau} G(u_k \given u_\tau)$, writing $G(u_k \given u_\tau)$ for the step-$k$ expected free energy of~\eqref{eq:efe-def} evaluated under the rollout for the whole policy $u_\tau$.
\end{proposition}
The proof is in Appendix \ref{app:proof_prop1}.


\subsubsection{Implementation}\label{sec:implementation}
The policy space is enumerated ($|\mathcal{U}|^K$ rollouts, vectorised over policies), and a single forward sweep through the transition and emission factors computes, for every $u_\tau$ and $k\in\tau$, the predictive beliefs, the goal cross-entropies $\E_{q(y_k|u_\tau)}\big[\ln p_*(y_k)\big]$, and the mutual informations $I\big[x_k,\Theta;\,y_k | u_\tau\big]$. Because of the form constraints (\autoref{def:form-constraint}), inference on the planning graph is exact and we need not iterate message passing. The only iterative computation is the scalar dual.

However, solving for $K$ KKT-multipliers is computationally demanding. Instead, we aggregate the per-step constraints~\eqref{eq:info-constraint} into one horizon constraint, 
\begin{equation}
	\E_{q(u_\tau)}\Big[\sum_{k} I\big[x_k,\Theta;\,y_k | u_\tau\big] \Big] \ge \beta \qquad \text{for}\ \beta \equiv \sum_{k}\beta_k \, ,
\end{equation} 
with a single multiplier $\gamma$ shared across the horizon. This is a relaxation rather than a reformulation: meeting every per-step floor implies the summed one, but not conversely. So the agent is free to over-inform at one step and leave another slack. 
The shared multiplier $\gamma$ traces a one-parameter family of policies containing the expected-free-energy policy at $\gamma = 1$, and it replaces a $K$-dimensional root-find by a $1$-dimensional one. The two constraints (per-step and aggregated) agree whenever the per-step multipliers coincide at the optimum, which they do when the information floors and the attainable information are exchangeable across the horizon. 

When solving the KKT-multiplier, we set an upper bound, $\gamma_{\max}$ ($10^{3}$ in all experiments). The expected information gain is monotonically non-decreasing in $\gamma$ within the exponential family generated by~\eqref{eq:policy-star} (\autoref{lem:gamma} in \refappx{appx:derivation}). So complementary slackness leaves two possibilities: if the constraint is already met at $\gamma = 0$ then it is inactive and the agent experiences no epistemic drive. Otherwise, it binds at some $\gamma \in (0,\gamma_{\max})$, where the expected information gain matches the sum of goal prior entropies. That interior root of~\eqref{eq:gamma-condition} has no closed form, the expected information gain being a ratio of sums of exponentials in $\gamma$. We therefore bracket it by the sign change of the binding condition over $[0,\gamma_{\max}]$ and locate it with a Newton iteration from the bracket midpoint, falling back to bisection on the same bracket whenever Newton leaves it. A third case arises when the floor exceeds the information any policy can supply: the constraint cannot be met and the multiplier grows without bound. Thus, we cap it at $\gamma_{\max}$, where the epistemic drive is maximal.

Acting requires a single decision rather than the full policy belief. The agent commits to the maximum a posteriori action of the stationary policy,
\begin{equation}\label{eq:map-action}
    \hat{u}_{t+1} \;=\; \arg\max_{u_{t+1}\in\mathcal{U}} q^\star(u_{t+1})\,,
\end{equation}
with $q^\star(u_{t+1})$ the first-step marginal of~\eqref{eq:policy-star}. It replans at every step.

\section{Experiments} \label{sec:experiments}

We run three experiments\footnote{Code available at \href{https://github.com/biaslab/IWAI2026-EFEasBetheConstraint}{https://github.com/biaslab/IWAI2026-EFEasBetheConstraint}.}: a canonical T-maze, an information-gathering cue grid, and a gate-cue task with an unknown transition.

All experiments run $200$ trials from fixed seeds with $\gamma_{\max} = 10^{3}$, except the floor sweep of \autoref{fig:floor-sweep}b, which runs $100$ trials per floor value. Throughout, $T$ denotes the length of a trial, $K$ the length of the planning horizon, and CBFE the constrained Bethe free energy agent of \autoref{sec:stationary_points}. Goal priors are uniform on every observation modality except reward, where $p_*$ places $(10^{-3},\, 0.998,\, 10^{-3})$ on (null, reward, loss). In \autoref{sec:gatecue} the Dirichlet prior over the unknown gate column has total concentration $4$ (uniform mean) against $50$ on the known transitions.

As baselines, we use a standard EFE minimisation planner and Q-MDP~\cite{littman1995learning}. Q-MDP is equipped with the true transition and emission, the regime in which the parameter posterior $q(\Theta)$ collapses to a point mass and the novelty term in~\eqref{eq:policy-star} vanishes. It acts by maximising the optimal Q-function of the underlying fully observable Markov decision process, averaged over its filtering belief. Because that Q-function assumes uncertainty resolves at no cost, Q-MDP does not gather information. We grant it a longer lookahead than the CBFE agent ($10$ against $6$ steps on the cue grid, $4$ against $3$ on the gate-cue task), so that a shorter horizon cannot account for the gap. The comparison thus isolates the information constraint~\eqref{eq:info-constraint} under a shared generative model.

\subsection{Canonical T-maze}\label{sec:tmaze}
The T-maze is a discrete benchmark in which information seeking is decisive \cite{da2020active}. The states are $\{\text{center}, \text{cue}, \text{right arm}, \text{left arm}\}$ and the reward is $c \in \{\text{right}, \text{left}\}$, drawn uniformly at random across trials. The action $u_t \in \mathcal{U} = \{0, 1, 2, 3\}$ encodes the target location. The observation space $\mathcal{Y}$ is the triple (location, reward, cue). Entering the arm matching $c$ emits a reward outcome with probability $\rho = 0.98$ and a loss outcome otherwise, while the worse arm flips these probabilities. Visiting the cue location deterministically reveals $c$. Each trial runs for $T = 3$ steps, and the planning horizon is set to the same length, $K = 3$.

\subsubsection{Results.}
\autoref{tab:tmaze} reports the full comparison. The information constraint~\eqref{eq:info-constraint} makes cue-seeking systematic: the constrained agent visits the cue on every trial and always reaches the rewarding arm in the minimum two steps. Q-MDP reaches the cue almost as often, but assigns no value to the information it supplies, so it does not visit it first and still gambles on the wrong arm in $18\%$ of trials. Beyond trajectory agreement, we verify \autoref{prop:efe} directly: running the constrained Bethe agent with $\gamma=1$ and a goal prior shared with the EFE agent, the two policy posteriors agree on every planning call of all $200$ trials to machine precision, $\max_{u_\tau} |q^\star(u_\tau) - q_{\mathrm{EFE}}(u_\tau)| \le 3 \cdot 10^{-16}$.

\begin{table}[htb]
    \centering
    \caption{Canonical T-maze over $200$ trials ($T = K = 3$, $\rho = 0.98$; $\pm$ one standard error over trials). The constrained Bethe and standard-EFE agents behave identically. Q-MDP, blind to future information gain, reaches the cue without exploiting it and gambles on an arm in the remaining trials.
    }
    \label{tab:tmaze}
    \setlength{\arrayrulewidth}{0.25mm}
    \setlength{\tabcolsep}{9pt}
    \renewcommand{\arraystretch}{1.2}
    \begin{tabular}{| l | c | c | c | c |}
        \hline
		& Reach rate & Cue-visit rate & $\langle$time-to-target$\rangle$ & $\langle$return$\rangle$ \\
        \hline
        Q-MDP             & $0.82 \pm .03$ & $0.79 \pm .03$ & $2.68 \pm .05$ & $+1.22 \pm .06$ \\
        EFE & $1.00$ & $1.00$ & $2.00$ & $+1.92 \pm .03$ \\
        CBFE      & $1.00$ & $1.00$ & $2.00$ & $+1.92 \pm .03$ \\
        \hline
    \end{tabular}
            \vspace{-10pt}
\end{table}

\subsection{Information-gathering grid experiment}\label{sec:cuegrid}
The hidden state factorises into a grid cell and a reward condition $c \in \{1, 2\}$, fixed within a trial and uniform across trials, which selects the rewarding one of two absorbing goal cells.
The action $u_t \in \mathcal{U} = \{\textsc{ccw}, \textsc{cw}, \textsc{forward}\}$ -- turn counter-clockwise, turn clockwise, or step forward -- drives deterministic egocentric navigation.
The observation $y_t \in \mathcal{Y}$ is the quadruple (cell, heading, cue, reward): cell and heading are reported deterministically, the rewarding goal emits a reward outcome and the other a loss. A cue cell at the end of a dead-end corridor reveals $c$ when occupied.
Unlike the T-maze, the cue is a multi-step detour away from both goals, so an agent that does not value information gain has no incentive to visit it.
Each trial runs for at most $T_{\max} = 25$ steps with planning horizon $K = 6$, and the goal prior is peaked on the reward outcome as in the T-maze.

\subsubsection{Results.}
The constrained Bethe agent visits the cue on every trial, detouring to resolve $c$ before committing, and reaches the rewarding goal on $84.0 \pm 2.6\%$ of trials (\autoref{fig:cuegrid}, left). The residual failures are horizon-induced dithering in the cue's dead-end corridor, from which no $6$-step rollout reaches a goal.
EFE never visits the cue because the information gain of the detour is outweighed by the goal cross-entropy cost of postponing the goal. So it rushes to the nearest goal and gambles, reaching the rewarding one on $45.0 \pm 3.5\%$ of trials with negative mean return.
Q-MDP likewise never values the off-path cue (a $1\%$ incidental visit rate) and reaches the rewarding goal on only $36.5 \pm 3.4\%$ (\autoref{fig:cuegrid}, right).
Re-encoding the preferences in the canonical loss-averse log-preference form (EFE-LA in \autoref{tab:cuegrid}) rescues the EFE planner: a blind gamble then has negative value relative to waiting, so it visits the cue on every trial and reaches the rewarding goal on all but one of the $200$, exceeding even the saturated constrained agent, which loses trials to post-cue dithering.
The comparison points to a design trade-off: standard EFE obtains exploration through preference design, with the loss-aversion magnitudes implicitly setting the explore--exploit balance, whereas the constrained agent demands it explicitly.

Planning cost is dominated by the rollout shared with EFE ($\approx 0.35$\,s per call). Solving the scalar dual adds under $3$\,ms.

\begin{figure}[tb]
    \centering
    \includegraphics[width=\textwidth]{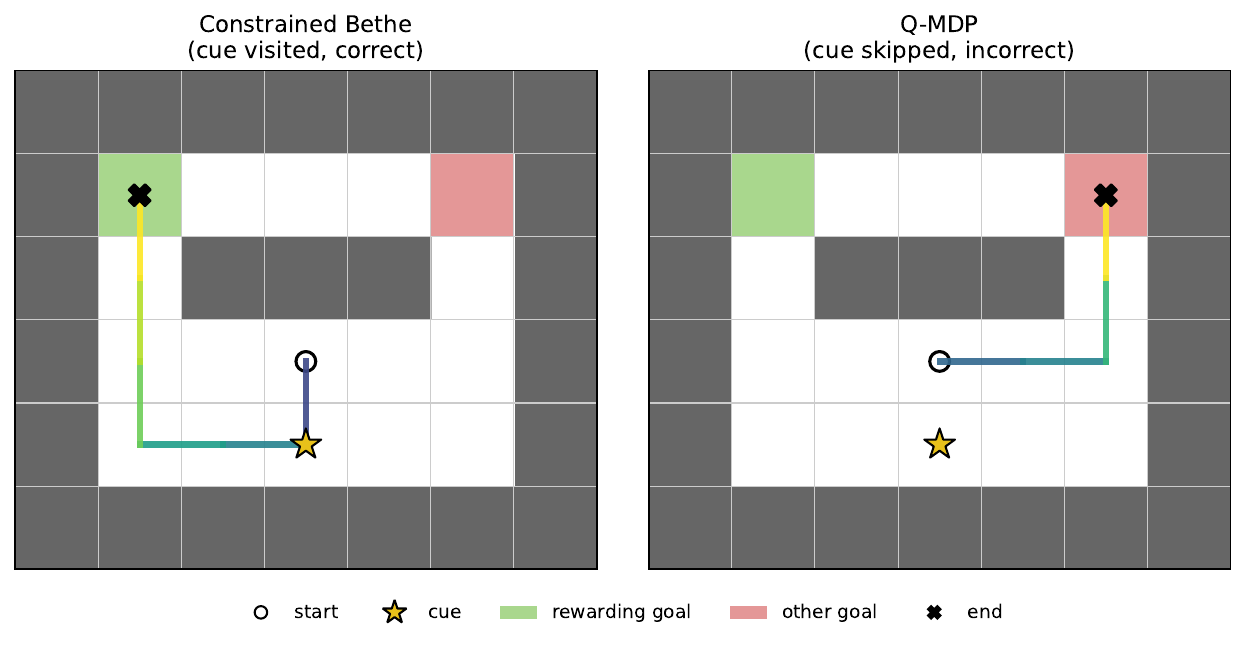}
            \vspace{-12pt}
    \caption{A representative cue-grid trial. Left: the constrained Bethe agent detours to the off-path cue (gold star) to resolve the reward condition, then proceeds to the rewarding goal (green). Right: Q-MDP ignores the cue and commits to a goal directly, here the non-rewarding one (red). Walls are grey, the open circle is the start, and paths shade from light (early) to dark (late).}
    \label{fig:cuegrid}
\end{figure}

\begin{table}[bt]
    \centering
    \caption{Information-gathering grid over $200$ trials ($K = 6$, $T_{\max} = 25$; $\pm$ one standard error). Reach rate is the fraction of trials reaching the rewarding goal. time-to-target is over the trials that reach a goal, which favours agents with lower reach rates. The unconditional return column is unaffected.
    }
    \label{tab:cuegrid}
    \setlength{\arrayrulewidth}{0.25mm}
    \setlength{\tabcolsep}{9pt}
    \renewcommand{\arraystretch}{1.2}
    \begin{tabular}{| l | c | c | c | c |}
        \hline
			 & Reach rate & Cue-visit rate & $\langle$time-to-target$\rangle$ & $\langle$return$\rangle$ \\
        \hline
        CBFE & $0.84 \pm .03$ & $1.00$ & $12.49 \pm 0.27$ & $+0.83 \pm .03$ \\
        EFE & $0.45 \pm .04$ & $0.00$ & $6.00 \pm 0.00$ & $-0.10 \pm .07$ \\
        EFE-LA & $1.00 \pm .01$ & $1.00$ & $12.32 \pm 0.18$ & $+1.00 \pm .01$ \\
        Q-MDP             & $0.37 \pm .03$ & $0.01 \pm .01$ & $13.08 \pm 0.34$ & $+0.02 \pm .06$ \\
        \hline
    \end{tabular}
        \vspace{-15pt}
\end{table}

\subsubsection{The floor selects the regime.}
In the benchmark above, $\beta$ is fixed to the full-joint goal-prior entropy, which no rollout attains, so the dual saturates at $\gamma_{\max}$. Basing $\beta$ on the reward modality alone and varying $H\big[p_*\big]$ instead exposes all three Karush--Kuhn--Tucker regimes of \autoref{lem:gamma} (\autoref{fig:floor-sweep}): the solved multiplier is zero while the floor is already met, rises through a unique interior root crossing the expected-free-energy value $\gamma = 1$, and saturates only once the floor exceeds the most informative rollout. Behaviourally the agent passes from a blind gamble to visiting the cue on every trial, so the multiplier is genuinely solved from the information demand.

\begin{figure}[htb]
    \centering
    \includegraphics[width=\textwidth]{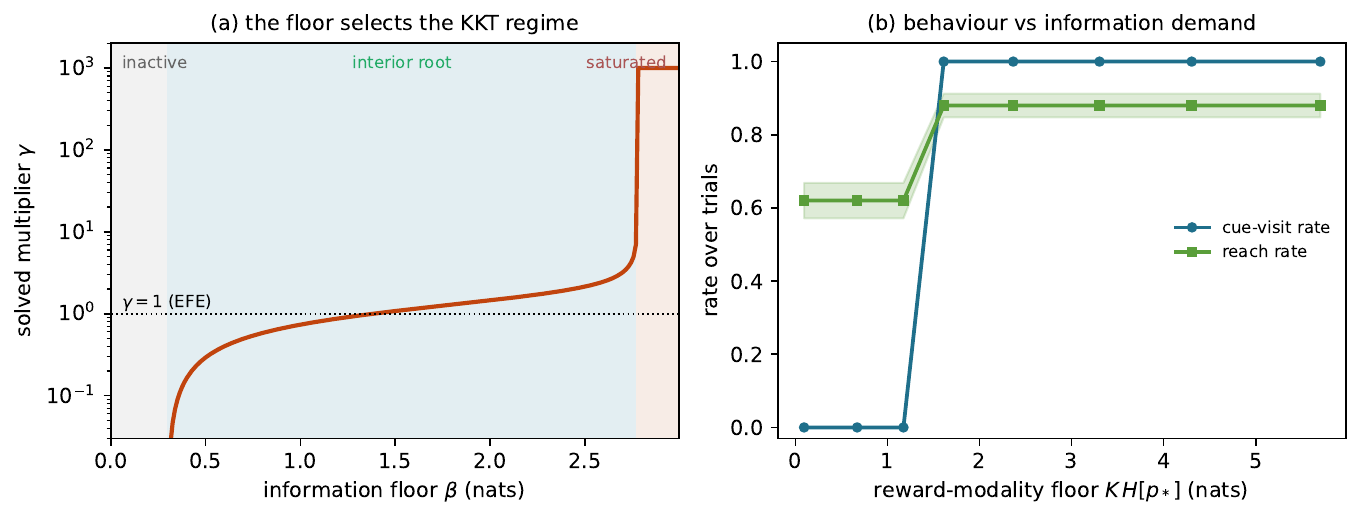}
	\vspace{-15pt}
    \caption{The information floor selects the dual's regime (cue grid). (a) Solved multiplier $\gamma$ against the per-horizon floor $\beta$ for a representative planning state: inactive at zero, a unique interior root crossing $\gamma = 1$, and saturation at $\gamma_{\max}$. (b) Cue-visit and rewarding-goal rates as the reward-modality floor $K\,H[p_*]$ is varied, rising from a blind gamble to maximal information seeking at the inactive--binding boundary of (a).}
    \label{fig:floor-sweep}
        \vspace{-10pt}
\end{figure}

\subsection{Salience--novelty allocation under learning}\label{sec:gatecue}
In the previous tasks, the dynamics were known and the only uncertainty was the hidden state. We now place the agent at a junction with two competing epistemic sources and let it learn. A hidden reward condition $c \in \{1, 2\}$, uniform per trial, is revealed at a \textsc{cue} (salience). In a separate direction lies a \textsc{gate} with an unknown transition. The agent holds a diffuse Dirichlet prior on the single column $B = p(x_t \mid x_{t-1} = \textsc{gate}, u_t = \textsc{use})$, whose true (fixed) destination can only be learned by probing. Resolving $c$ contributes salience $I\big[x_k;\,y_k \mid \Theta, u_\tau \big]$, and probing the gate contributes novelty $I\big[B;\,y_k \mid u_\tau \big]$, the parameter term reducing to $B$ since $A$ is known here. The goals sit on two one-way three-step branches sharing only the start, so with horizon $K = 3$ the agent commits to one source per trial. Crucially, the Dirichlet posterior over $B$ persists across trials, so what the agent learns about the gate carries over. We run $200$ trials.

\subsubsection{Results.}
Early on, while $B$ is uncertain, novelty exceeds the cue salience and the single multiplier $\gamma$ drives the agent down the gate branch. For the first $29$ trials it probes the gate every time and, with $c$ unresolved, gambles on a goal ($34\%$ correct). As probing concentrates the Dirichlet posterior (gate-belief error falling from $0.83$ to $0.18$) novelty decays, and under the same constraint the agent's effort shifts to the cue. It visits the cue on every subsequent trial, resolves $c$, and reaches the rewarding goal on $100\%$ of them (\autoref{fig:gatecue}a,b).

Two controls isolate the mechanism: a salience-only agent given the true gate (point-mass $B$, novelty zero) never probes and heads straight for the cue from the first trial. An agent whose Dirichlet posterior is reset each trial keeps probing indefinitely. Q-MDP, valuing neither source, performs at chance ($49\%$).

A novelty-equipped standard EFE planner shows the same shift (\autoref{fig:gatecue}b) but switches earlier (last probe at trial $15$; $95.5\%$ correct overall vs.\ $90.5\%$ for the saturated agent), the unit multiplier weighting information more weakly than the saturated dual. In both cases the mechanism is carried by the joint state--parameter information gain: a single multiplier reallocates effort from novelty to salience as the agent learns its own dynamics.

\begin{figure}[thb]
    \centering
    \includegraphics[width=\textwidth]{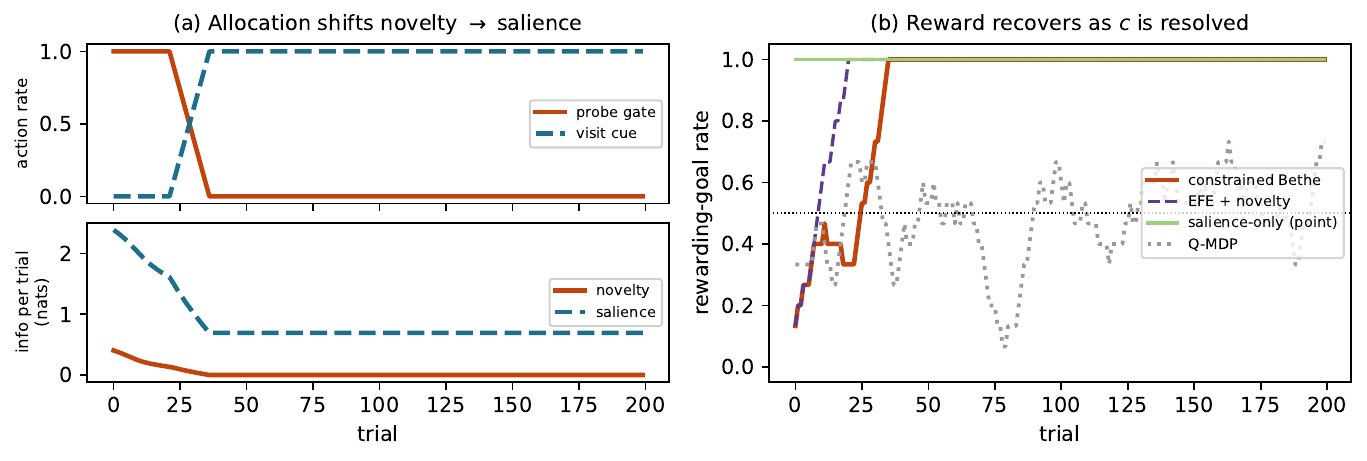}
            \vspace{-15pt}
    \caption{Salience--novelty allocation under learning. \textbf{(a)} The gate-probe rate falls and the cue-visit rate rises across trials (top), tracked by novelty and salience (bottom). \textbf{(b)} The reward rate recovers to $1$ as the agent stops gambling; EFE+novelty planner shifts earlier, the salience-only agent is always correct, and Q-MDP is at chance.}
    \label{fig:gatecue}
    \vspace{-5pt}
\end{figure}

\section{Discussion} \label{sec:discussion}
EFE-based planning can be cast as variational inference on a model augmented with preference and epistemic priors~\cite{de2025expected}. Epistemic priors realise information-seeking as a further model factor with a fixed weight, keeping the objective an ordinary variational free energy at the price of folding the epistemic drive into the agent's generative model. We realise it as an inequality constraint with a solved dual: the drive switches off, binds, or saturates by complementary slackness (\autoref{lem:gamma}). Constraints on the Lagrangian are a powerful way to enrich inference algorithms \cite{csenoz2021variational}. Van de Laar et al., for example, cast planning as constrained Bethe minimisation through the same local constraint manipulations, but bound the probability of violating an outcome specification -- effectively a safety requirement on the goal \cite{van2021chance}. Both the information and the chance constraint are examples of encoding desired outcomes of the inference process, which, in our view, should not be viewed as a priori knowledge.


\section{Conclusion} \label{sec:conclusion}

We posed action selection as constrained minimisation of a Bethe free energy on a goal-augmented planning graph. A single information constraint per step, a lower bound on the expected mutual information with a Karush--Kuhn--Tucker multiplier that is solved rather than tuned, yields a one-parameter family of policies containing the expected-free-energy policy exactly at unit multiplier. The constrained Bethe agent matches a standard expected-free-energy planner to machine precision on the T-maze, detours to an off-path cue that an exploit-only Q-MDP baseline ignores, and reallocates its effort from novelty- to salience-seeking as it learns an unknown transition.

\begin{credits}
 \subsubsection{\ackname} This work is supported by the Sector Plan Techniek of the Ministry of Education,
Culture and Science of the Netherlands (OCW).

 \subsubsection{\discintname}
The author has no competing interests to declare that are
 relevant to the content of this article.
\end{credits}

\bibliographystyle{splncs04}
\bibliography{references}

@book{boyd2004convex,
  title={Convex optimization},
  author={Boyd, Stephen and Vandenberghe, Lieven},
  year={2004},
  publisher={Cambridge University Press}
}

@article{champion2021realizing,
  title={Realizing active inference in variational message passing: The outcome-blind certainty seeker},
  author={Champion, Th{\'e}ophile and Grze{\'s}, Marek and Bowman, Howard},
  journal={Neural Computation},
  volume={33},
  number={10},
  pages={2762--2826},
  year={2021},
  publisher={MIT Press}
}

@book{cover1999elements,
  title={Elements of information theory},
  author={Cover, Thomas M},
  year={1999},
  publisher={John Wiley \& Sons}
}

@article{csenoz2021variational,
  title={Variational message passing and local constraint manipulation in factor graphs},
  author={{\c{S}}en{\"o}z, {\.I}smail and van de Laar, Thijs and Bagaev, Dmitry and de Vries, Bert},
  journal={Entropy},
  volume={23},
  number={7},
  year={2021},
  publisher={MDPI}
}

@article{da2020active,
  title={Active inference on discrete state-spaces: A synthesis},
  author={Da Costa, Lancelot and Parr, Thomas and Sajid, Noor and Veselic, Sebastijan and Neacsu, Victorita and Friston, Karl},
  journal={Journal of Mathematical Psychology},
  volume={99},
  pages={102447},
  year={2020},
  publisher={Elsevier}
}

@article{de2017factor,
  title={A factor graph description of deep temporal active inference},
  author={De Vries, Bert and Friston, Karl J},
  journal={Frontiers in Computational Neuroscience},
  volume={11},
  number={95},
  year={2017},
}

@article{de2025expected,
  title={Expected Free Energy-based Planning as Variational Inference},
  author={de Vries, Bert and Nuijten, Wouter and van de Laar, Thijs and Kouw, Wouter and Adamiat, Sepideh and Nisslbeck, Tim and Lukashchuk, Mykola and Nguyen, Hoang Minh Huu and Araya, Marco Hidalgo and Tresor, Raphael and Jenneskens, Thijs and Nikoloska, Ivana and Ganapathy
Subramanian, Raaja and van Erp, Bart and Bagaev, Dmitry and Podusenko, Albert},
  journal={arXiv:2504.14898},
  year={2025}
}

@article{friston2025pixels,
  title={From pixels to planning: scale-free active inference},
  author={Friston, Karl and Heins, Conor and Verbelen, Tim and Da Costa, Lancelot and Salvatori, Tommaso and Markovic, Dimitrije and Tschantz, Alexander and Koudahl, Magnus and Buckley, Christopher and Parr, Thomas},
  journal={Frontiers in Network Physiology},
  volume={5},
  pages={1521963},
  year={2025},
  publisher={Frontiers Media SA}
}

@inproceedings{kouw2025message,
  title={Message passing-based inference in an autoregressive active inference agent},
  author={Kouw, Wouter M and Nisslbeck, Tim N and Nuijten, Wouter LN},
  booktitle={International Workshop on Active Inference},
  pages={285--298},
  year={2025},
  organization={Springer}
}

@article{ladosz2022exploration,
  title={Exploration in deep reinforcement learning: A survey},
  author={Ladosz, Pawel and Weng, Lilian and Kim, Minwoo and Oh, Hyondong},
  journal={Information Fusion},
  volume={85},
  pages={1--22},
  year={2022},
  publisher={Elsevier}
}

@article{loeliger2007factor,
  title={The factor graph approach to model-based signal processing},
  author={Loeliger, Hans-Andrea and Dauwels, Justin and Hu, Junli and Korl, Sascha and Ping, Li and Kschischang, Frank R},
  journal={Proceedings of the IEEE},
  volume={95},
  number={6},
  pages={1295--1322},
  year={2007},
  publisher={IEEE}
}

@inproceedings{nuijten2025message,
  title={A Message Passing Realization of Expected Free Energy Minimization},
  author={Nuijten, Wouter WL and Lukashchuk, Mykola and van de Laar, Thijs and de Vries, Bert},
  booktitle={International Workshop on Active Inference},
  pages={75-98},
  year={2025},
}

@book{parr2022active,
  title={Active inference: the free energy principle in mind, brain, and behavior},
  author={Parr, Thomas and Pezzulo, Giovanni and Friston, Karl J},
  year={2022},
  publisher={MIT Press}
}

@book{pearl1988probabilistic,
  title={Probabilistic reasoning in intelligent systems: networks of plausible inference},
  author={Pearl, Judea},
  year={1988},
  publisher={Elsevier}
}

@article{van2021chance,
  title={Chance-constrained active inference},
  author={van de Laar, Thijs and {\c{S}}en{\"o}z, {\.I}smail and {\"O}z{\c{c}}elikkale, Ay{\c{c}}a and Wymeersch, Henk},
  journal={Neural Computation},
  volume={33},
  number={10},
  pages={2710--2735},
  year={2021},
  publisher={MIT Press}
}

@article{van2022active,
  title={Active Inference and Epistemic Value in Graphical Models},
  author={van de Laar, Thijs and Koudahl, Magnus and van Erp, Bart and de Vries, Bert},
  journal={Frontiers in Robotics and AI},
  volume={9},
  year={2022},
  publisher={Frontiers}
}

@article{van2024realizing,
  title={Realizing Synthetic Active Inference Agents, Part {II}: Variational Message Updates},
  author={van de Laar, Thijs and Koudahl, Magnus and de Vries, Bert},
  journal={Neural Computation},
  volume={37},
  number={1},
  pages={38--75},
  year={2024},
  publisher={MIT Press}
}

@article{yedidia2005constructing,
  title={Constructing free-energy approximations and generalized belief propagation algorithms},
  author={Yedidia, Jonathan S and Freeman, William T and Weiss, Yair},
  journal={IEEE Transactions on Information Theory},
  volume={51},
  number={7},
  pages={2282--2312},
  year={2005}
}

@article{zhang2021unifying,
  title={Unifying message passing algorithms under the framework of constrained {Bethe} free energy minimization},
  author={Zhang, Dan and Song, Xiaohang and Wang, Wenjin and Fettweis, Gerhard and Gao, Xiqi},
  journal={IEEE Transactions on Wireless Communications},
  volume={20},
  number={7},
  pages={4144--4158},
  year={2021},
  publisher={IEEE}
}

@inproceedings{littman1995learning,
  title={Learning policies for partially observable environments: Scaling up},
  author={Littman, Michael L. and Cassandra, Anthony R. and Kaelbling, Leslie Pack},
  booktitle={Machine Learning Proceedings 1995},
  pages={362--370},
  year={1995},
  publisher={Elsevier}
}

@article{schwobel2018active,
  title={Active inference, belief propagation, and the {B}ethe approximation},
  author={Schw{\"o}bel, Sarah and Kiebel, Stefan and Markovi{\'c}, Dimitrije},
  journal={Neural Computation},
  volume={30},
  number={9},
  pages={2530--2567},
  year={2018},
  publisher={MIT Press}
}

\appendix
\section[Appendix: Proofs]{Appendix: proofs of theoretical results}\label{appx:derivation}

\paragraph{Planning graph and constrained family.}
Let $V$ denote the set of cluster (factor) nodes in the planning graph (\autoref{fig:ffg-planning}) and $E$ the set of edges (shared variables). For~\eqref{eq:generative-planning-model} the cluster set is
\begin{equation}
    V \;=\; \{x_t,\,A,\,B\} \,\cup\, \{\sfx_k,\,\sfy_k,\,\sfu_k,\,\sfg_k\}_{k\in\tau}\,,
\end{equation}
where each composite cluster is named, in sans serif, after the variable it emits -- $\sfx_k$ for the transition, $\sfy_k$ for the emission, $\sfu_k$ for the control, and $\sfg_k$ for the goal cluster at time $k$ -- while the singleton prior clusters are named by their variables. 
The edge set is $E = \{A,B\}\cup\{x_k\}_{k\in\{t\}\cup\tau}\cup\{u_k,y_k\}_{k\in\tau}$, with degrees as in \autoref{eq:bethe-def}.

The form constraints of \autoref{def:form-constraint} are local constraint manipulations in the sense of Şenöz et al.~\cite{csenoz2021variational}. Fixing $q(x_t), q(B), q(A)$ is a structured mean-field factorisation between the parameters and the trajectory with the parameter beliefs held at their past-graph values. Joining the control clusters $\{\sfu_k\}_{k\in\tau}$ into a single policy cluster with belief $q(u_\tau)$ is a cluster (region) choice, and pinning the rollout conditionals is a form constraint on the conditional beliefs. Jointly they restrict the variational family to
\begin{align}
    &q(y_\tau, x_{t:t+K}, u_\tau, \Theta) = \label{eq:constrained-family} \\
    &\qquad q(u_\tau)\; p(x_t\given\mathcal{D}_t)\, p(B\given\mathcal{D}_t)\, p(A\given\mathcal{D}_t) \prod_{k\in\tau} p(x_k \given x_{k-1}, u_k, B)\, p(y_k \given x_k, A)\,, \nonumber
\end{align}
in which the policy belief $q(u_\tau)$ is the only free factor, and all cluster and edge beliefs of~\eqref{eq:bethe-def} are the corresponding marginals of~\eqref{eq:constrained-family}.
	
\subsection{Proof of \autoref{lem:collapse}} \label{app:proof_lem1}
\begin{proof}
Conditional on $u_\tau$ and with the parameter beliefs fixed, the family~\eqref{eq:constrained-family} factorises along a tree: the parameter-sharing equality chains of \autoref{fig:ffg-planning} enter only through the fixed beliefs $q(B), q(A)$, which severs every cycle. On a tree the Bethe entropy decomposition is exact~\cite{yedidia2005constructing}, so the Bethe free energy~\eqref{eq:bethe-def}, evaluated on the marginals of~\eqref{eq:constrained-family} with the controls collected in the policy cluster, equals the variational free energy of the joint,
\begin{equation}
    \mathcal{B}[q] \;=\; \E_{q}\Big[\ln \frac{q(y_\tau, x_{t:t+K}, u_\tau, \Theta)}{p(y_{\tau}, u_\tau, x_{\tau}, x_t, B, A \given \mathcal{D}_t)}\Big]\,.
\end{equation}
Substituting~\eqref{eq:constrained-family} and the goal-augmented model~\eqref{eq:generative-planning-model}, every factor of~\eqref{eq:constrained-family} cancels against its counterpart in $p$ -- the priors $p(x_t\given\mathcal{D}_t), p(B\given\mathcal{D}_t), p(A\given\mathcal{D}_t)$ exactly, and the pinned transition and emission conditionals pairwise -- leaving
\begin{equation}
    \ln \frac{q(\cdot)}{p(\cdot)} \;=\; \ln \frac{q(u_\tau)}{p(u_\tau)} \;-\; \sum_{k\in\tau} \ln p_*(y_k)\,,
\end{equation}
up to the additive normalisation constant of $p$. Taking the expectation under~\eqref{eq:constrained-family} yields~\eqref{eq:bethe-collapse}. \qed
\end{proof}

\subsection{Proof of \autoref{thm:policy}} \label{app:proof_thm1}
\begin{proof}
The form constraints make $\sum_{u_\tau} q(u_\tau)\, I\big[x_k,\Theta;\,y_k\given u_\tau \big]$ linear in $q(u_\tau)$ with fixed coefficients.
By \autoref{lem:collapse}, the variation of the Lagrangian~\eqref{eq:lagrangian} with respect to $q(u_\tau)$ is
\begin{equation}
\begin{split}
	\frac{\delta\mathcal{L}}{\delta q(u_\tau)}=&\,  \ln q(u_\tau) + 1 - \ln p(u_\tau) + \lambda \\
	&-  \sum_{k\in\tau} \Big(\E_{q(y_k | u_\tau)}\big[ \ln p_*(y_k)\big] + \gamma_k I\big[x_k,\Theta;\,y_k | u_\tau\big]\Big) .
\end{split}
\end{equation}
Setting it to $0$, solving for $\ln q(u_\tau)$ and fixing $\lambda$ by normalisation gives~\eqref{eq:policy-star}. The control marginals follow by marginalising the policy cluster. The signs and feasibility of the duals $\gamma_k \ge 0$ are the Karush--Kuhn--Tucker conditions for the inequalities~\eqref{eq:info-constraint}, resolved by \autoref{lem:gamma} below. \qed
\end{proof}

\subsection{Proof of \autoref{prop:efe}} \label{app:proof_prop1}
\begin{proof}
Per time step in the planning horizon, the identity~\eqref{eq:efe-identity} evaluated under the rollout reads $-G(u_k \given u_\tau) = \E_{q(y_k\given u_\tau)}\big[\ln p_*(y_k) \big] + I\big[x_k,\Theta;\,y_k\given u_\tau \big]$. Setting $\gamma_k = 1$ for all $k\in\tau$ in~\eqref{eq:policy-star} therefore gives
\begin{align} 
q^\star(u_\tau) &\propto p(u_\tau)\exp\Big(\sum_k -G(u_k\given u_\tau)\Big) \\
&= p(u_\tau)\exp\big(-G(u_\tau)\big) \, ,
\end{align} 
which is $q_{\mathrm{EFE}}(u_\tau)$, the solution under the EFE functional. \qed
\end{proof}

\paragraph{Determining the dual.}
As described in \autoref{sec:implementation}, the implementation aggregates the per-step constraints into a single horizon constraint. Writing $S(u_\tau) = \sum_{k\in\tau} I\big[x_k,\Theta;\,y_k \given u_\tau\big]$ for the information a rollout supplies, that constraint reads $\E_{q(u_\tau)}\big[S(u_\tau)\big] \ge \beta = \sum_{k\in\tau}\beta_k$, with one multiplier $\gamma$ shared across the horizon, which preserves the Karush--Kuhn--Tucker structure with a single dual variable. The following lemma characterises its solution. The per-step system is the same object in $K$ dimensions: writing
\[
    \ln Z(\gamma_1,\dots,\gamma_K) \;=\; \ln \sum_{u_\tau} \bar{p}(u_\tau)\,\exp\Big(\sum_{k\in\tau} \gamma_k\, I\big[x_k,\Theta;\,y_k \given u_\tau\big]\Big) \, ,
\]
the stationarity conditions of the per-step problem read $\nabla \ln Z = (\beta_1,\dots,\beta_K)$, whose Jacobian is the covariance matrix of the per-step information terms under $q^\star$ and is therefore positive semi-definite. So $\ln Z$ is convex and $\nabla \ln Z$ is monotone, and the lemma below is its restriction to the diagonal $\gamma_k \equiv \gamma$, along which the Jacobian contracts to the scalar variance in~\eqref{eq:gamma-derivative}.

\begin{lemma}[Information multiplier]\label{lem:gamma}
Let the goal-tilted base measure be
\begin{equation}\label{eq:base-measure}
\bar{p}(u_\tau) \propto p(u_\tau)\exp\Big(\sum_k \E_{q(y_k|u_\tau)}\big[\ln p_*(y_k) \big] \Big) \, ,
\end{equation}
and let
\begin{equation}
q^\star(u_\tau;\gamma) \propto \bar{p}(u_\tau)\exp\big(\gamma\, S(u_\tau) \big)
\end{equation}
be the measure that depends on the shared multiplier, with $S$ as above. The expected information gain $\E_{q^\star(u_\tau;\gamma)}\big[S \big]$ is non-decreasing in $\gamma$, rising from $\E_{\bar{p}}\big[S \big]$ at $\gamma=0$ to $\max_{u_\tau} S(u_\tau)$ as $\gamma\to\infty$.

The Karush--Kuhn--Tucker conditions admit exactly three cases:
\begin{enumerate}
\item if $\E_{\bar{p}}\big[S \big]\ge\beta$ the constraint is slack and $\gamma=0$;
\item if $\E_{\bar{p}}\big[S \big]<\beta\;<\;\max_{u_\tau} S(u_\tau)$ the binding condition
\begin{equation}\label{eq:gamma-condition}
    \E_{q^\star(u_\tau;\gamma)}\big[S(u_\tau)\big] \;=\; \beta
\end{equation}
has a unique root $\gamma>0$;
\item if $\beta\;\ge\;\max_{u_\tau} S(u_\tau)$ no policy can clear the floor and $\gamma$ saturates at the bracket boundary $\gamma_{\max}$, concentrating the policy on the most informative rollouts.
\end{enumerate}
\end{lemma}

\begin{proof}
Let
\begin{equation}\label{eq:log-partition}
    \ln Z(\gamma) \;=\; \ln \sum_{u_\tau} \bar{p}(u_\tau)\,\exp\big(\gamma\, S(u_\tau)\big)
\end{equation}
be its log-partition function, so that $q^\star(u_\tau;\gamma) = \bar{p}(u_\tau)\exp\big(\gamma S(u_\tau) - \ln Z(\gamma)\big)$.
The family $q^\star(u_\tau;\gamma)$ is a one-parameter exponential family in the natural parameter $\gamma$ with sufficient statistic $S$ and base measure $\bar{p}$, so the derivatives of $\ln Z$ return the cumulants of $S$,
\begin{equation}\label{eq:gamma-derivative}
    g(\gamma) \;\equiv\; \frac{\rmd \ln Z}{\rmd\gamma} \;=\; \E_{q^\star(u_\tau;\gamma)}\big[S\big] \, , \qquad \frac{\rmd g}{\rmd\gamma} \;=\; \mathrm{Var}_{q^\star(u_\tau;\gamma)}\big[S\big] \;\ge\; 0\,.
\end{equation}
The expected information gain is thus $g$, non-decreasing because $\ln Z$ is convex. Its endpoints are the ones claimed. At $\gamma = 0$ the tilt is absent, leaving $g(0) = \E_{\bar{p}}[S]$. For the other endpoint let $S_{\max}$ be the largest value $S$ takes on the support of $\bar{p}$ -- which is the support of the control prior, the goal tilt in~\eqref{eq:base-measure} being strictly positive -- and let $\mathcal{U}^\star$ collect the rollouts attaining it. For $u_\tau \notin \mathcal{U}^\star$ and $u^\star_\tau \in \mathcal{U}^\star$,
\begin{equation}
    \frac{q^\star(u_\tau;\gamma)}{q^\star(u^\star_\tau;\gamma)} \;=\; \frac{\bar{p}(u_\tau)}{\bar{p}(u^\star_\tau)}\,\exp\big(\!-\gamma\,[S_{\max} - S(u_\tau)]\big) \;\longrightarrow\; 0 \quad \text{as } \gamma\to\infty\,,
\end{equation}
since the bracket is strictly positive. The policy space is finite, so the mass of the finitely many non-maximising rollouts vanishes, $q^\star(\cdot\,;\gamma)$ concentrates on $\mathcal{U}^\star$, and $g(\gamma)\to S_{\max}$.

But monotonicity alone does not yield uniqueness. Strictness does. By~\eqref{eq:gamma-derivative}, $\rmd g/\rmd\gamma$ vanishes only if $S$ is constant on the support of $\bar{p}$. In that degenerate case $\E_{\bar{p}}[S] = S_{\max}$, the interval $\big(\E_{\bar{p}}[S], S_{\max}\big)$ is empty and case 2 cannot arise. Otherwise $\rmd g/\rmd\gamma > 0$ and $g$ is a continuous, strictly increasing bijection from $[0,\infty)$ onto $\big[\E_{\bar{p}}[S], S_{\max}\big)$. Note that $S_{\max}$ itself is approached but never attained at finite $\gamma$.

Besides stationarity, which produces the form $q^\star(u_\tau;\gamma)$, the Karush--Kuhn--Tucker conditions are
\begin{equation}
    g(\gamma) \;\ge\; \beta \, , \qquad \gamma \;\ge\; 0 \, , \qquad \gamma\,\big[\beta - g(\gamma)\big] \;=\; 0 \, ,
\end{equation}
that is, primal feasibility, dual feasibility and complementary slackness. The three cases are the trichotomy of $\beta$ against the range of $g$.
\begin{enumerate}
\item If $\beta \le g(0)$, then $\gamma = 0$ is already feasible and satisfies slackness. The constraint is inactive and the agent carries no epistemic drive.
\item If $g(0) < \beta < S_{\max}$, then $\gamma = 0$ is infeasible, so slackness forces $g(\gamma) = \beta$, which is~\eqref{eq:gamma-condition}. As $\beta$ lies in the range of $g$ and $g$ is strictly increasing there, this root exists and is unique.
\item If $\beta \ge S_{\max} > g(0)$, then $g(\gamma) < \beta$ at every finite $\gamma$: no policy clears the floor, the feasible set is empty and no Karush--Kuhn--Tucker point exists. The multiplier grows without bound and the implementation truncates it at $\gamma_{\max}$, where $q^\star$ is concentrated on $\mathcal{U}^\star$ -- the maximal information-seeking limit.
\end{enumerate}
\qed
\end{proof}

\end{document}